\documentclass[11pt,a4paper,twoside]{article}
\pdfoutput=1

\usepackage{amsmath}
\usepackage{amsthm}
\usepackage{microtype}
\usepackage{amssymb}
\usepackage[margin=2.5cm]{geometry}
\usepackage{tikz}

\author{Ralph C. Bottesch\thanks{Department of Computer Science, University of Innsbruck, Austria. This work was supported by the ERC Consolidator Grant QPROGRESS 615307 for the majority of its duration (while the author was a post-doc at CWI, Amsterdam), and by the Austrian Science Fund (FWF) project Y757 at the time of publication.}}
\title{On $\mathsf{W}[1]$-Hardness as Evidence for Intractability}

\theoremstyle{definition}
\newtheorem{theor}{Theorem}
\theoremstyle{definition}
\newtheorem{prop}[theor]{Proposition}
\theoremstyle{definition}
\newtheorem{lemma}[theor]{Lemma}
\theoremstyle{definition}
\newtheorem{corol}[theor]{Corollary}
\theoremstyle{definition}

\theoremstyle{definition}
\newtheorem{fact}[theor]{Fact}
\theoremstyle{definition}

\theoremstyle{definition}
\newtheorem{defin}[theor]{Definition}
\theoremstyle{definition}

\theoremstyle{definition}

\begin{document}
\maketitle
\begin{abstract}
The central conjecture of parameterized complexity states that $\mathsf{FPT}\neq\mathsf{W}[1]$, and is generally regarded as the parameterized counterpart to $\mathsf{P}\neq\mathsf{NP}$. We revisit the issue of the plausibility of $\mathsf{FPT}\neq\mathsf{W}[1]$, focusing on two aspects: the difficulty of proving the conjecture (assuming it holds), and how the relation between the two classes might differ from the one between $\mathsf{P}$ and $\mathsf{NP}$.

Regarding the first aspect, we give new evidence that separating $\mathsf{FPT}$ from $\mathsf{W}[1]$ would be considerably harder than doing the same for $\mathsf{P}$ and $\mathsf{NP}$. Our main result regarding the relation between $\mathsf{FPT}$ and $\mathsf{W}[1]$ states that the closure of $\mathsf{W[1]}$ under relativization with $\mathsf{FPT}$-oracles is precisely the class $\mathsf{W[P]}$, implying that either $\mathsf{FPT}$ is not low for $\mathsf{W}[1]$, or the $\mathsf{W}$-Hierarchy collapses. This theorem also has consequences for the $\mathsf{A}$-Hierarchy (a parameterized version of the Polynomial Hierarchy), namely that unless $\mathsf{W[P]}$ is a subset of some level $\mathsf{A}[t]$, there are structural differences between the $\mathsf{A}$-Hierarchy and the Polynomial Hierarchy. We also prove that under the unlikely assumption that $\mathsf{W[P]}$ collapses to $\mathsf{W}[1]$ in a specific way, the collapse of any two consecutive levels of the $\mathsf{A}$-Hierarchy implies the collapse of the entire hierarchy to a finite level; this extends a result of Chen, Flum, and Grohe (2005).

Finally, we give weak (oracle-based) evidence that the inclusion $\mathsf{W}[t]\subseteq\mathsf{A}[t]$ is strict for $t>1$, and that the $\mathsf{W}$-Hierarchy is proper. The latter result answers a question of Downey and Fellows (1993).
\end{abstract}

\section{Introduction}
The central conjecture of parameterized complexity theory states that $\mathsf{FPT}\neq \mathsf{W}[1]$. The complexity class $\mathsf{FPT}$ is a generalization of $\mathsf{P}$, and it also contains this class in the sense that regardless of which parameter we associate with the instances of a problem in $\mathsf{P}$, the resulting \emph{parameterized problem} is in $\mathsf{FPT}$. This inclusion is strict, as $\mathsf{FPT}$ also contains parameterized versions of problems that are provably not in $\mathsf{P}$. The class $\mathsf{W}[1]$ can be regarded as a parameterized counterpart to $\mathsf{NP}$. It can be defined in different ways, all of them quite technical, but the most common definition is in terms of a parameterized version of a particular $\mathsf{NP}$-complete problem (much like $\mathsf{NP}$ can be defined in terms of a Boolean circuit satisfiability problem). However, $\mathsf{W}[1]$ is not known or believed to contain \emph{all} parameterized versions of problems in $\mathsf{NP}$, and by defining complexity classes in terms of parameterizations of other $\mathsf{NP}$-complete problems, one actually obtains a large set of seemingly distinct parameterized analogues of $\mathsf{NP}$, some of which we list here:
\[
\mathsf{W}[1](=\mathsf{A}[1])\subseteq\mathsf{W}[2]\subseteq\ldots\mathsf{W}[t]\ldots\subseteq\mathsf{W[P]}\subseteq\mathsf{para\textrm{-}NP}.
\]

Among these, the most interesting classes are $\mathsf{W}[1]$ (a.k.a.\ $\mathsf{A}[1]$) and $\mathsf{W[P]}$, due to having many natural complete problems.

The basic intuition for why $\mathsf{W}[1]$ (and hence all classes in the above sequence) should differ from $\mathsf{FPT}$ is the same as for $\mathsf{P}\neq\mathsf{NP}$, namely that we do not know of any way to efficiently simulate nondeterministic computations deterministically. This intuition is often used to justify considering the $\mathsf{W}[1]$-hardness of a problem as evidence for its intractability. But because $\mathsf{FPT}$ is strictly larger than $\mathsf{P}$, while $\mathsf{W}[1]$ does not appear to capture all of the complexity of $\mathsf{NP}$, it seems that proving the central conjecture of parameterized complexity theory may be harder than separating $\mathsf{P}$ and $\mathsf{NP}$. We investigate qualitative differences between the two conjectures, as well as the more general question of whether $\mathsf{FPT}$ occupies the same place within $\mathsf{W}[1]$ as $\mathsf{P}$ does within $\mathsf{NP}$. We start by giving a brief summary of some relevant prior results.

That the central parameterized conjecture is at least as strong as its classical counterpart is easy to prove: If $\mathsf{NP}=\mathsf{P}$, then, as noted above, every parameterized version of every problem in $\mathsf{NP}$($=\mathsf{P}$) must be in $\mathsf{FPT}$, hence $\mathsf{W}[1]=\mathsf{FPT}$ (and, in fact, $\mathsf{para\textrm{-}NP}=\mathsf{W[P]}=\ldots=\mathsf{FPT}$). Thus we have that $\mathsf{FPT}\neq\mathsf{W}[1]\Rightarrow \mathsf{P}\neq\mathsf{NP}$. The converse of this implication is not known to hold, but Downey and Fellows \cite{df1} were the first to observe that a collapse of $\mathsf{W}[1]$ to $\mathsf{FPT}$ would at least imply the existence algorithms with sub-exponential running time for the $\mathsf{NP}$-complete problem $\textsc{3Sat}$. This would contradict the \emph{Exponential Time Hypothesis (ETH)}, first introduced by Imagliazzo, Paturi, and Zane \cite{ipz}, which states that for some constant $c>0$, \textsc{3Sat} can not be solved in time $O^\ast(2^{cn})$ by deterministic Turing machines (TMs). This conjecture has enjoyed much popularity recently, because, assuming ETH, for many problems it is possible to prove a complexity lower bound that matches that of the best known algorithm up to lower-order factors (see \cite{lokmarsau} for a survey of such results). Nevertheless, one should keep in mind that ETH is a much stronger statement than $\mathsf{P}\neq\mathsf{NP}$, since it rules out not only the existence of polynomial-time algorithms for $\textsc{3Sat}$, but also of those that run in up to exponential-time (for some bases). Putting all of these facts together, we have:
\[
\textrm{ETH}\Longrightarrow\mathsf{FPT}\neq\mathsf{W}[1]\Longrightarrow\ldots\Longrightarrow\mathsf{FPT}\neq\mathsf{W[P]}\Longrightarrow\mathsf{FPT}\neq\mathsf{para\textrm{-}NP}\Longrightarrow\mathsf{P}\neq\mathsf{NP}.
\]
The above sequence relates parameterized complexity conjectures to two classical ones, but it does not say which of them are closer in strength to ETH and which are closer to $\mathsf{P}\neq\mathsf{NP}$. The only known fact here is that $\mathsf{FPT}\neq\mathsf{para\textrm{-}NP}\Leftrightarrow \mathsf{P}\neq\mathsf{NP}$ (see \cite[Corollary 2.13]{fg}), but there is strong evidence suggesting that all of the other parameterized conjectures listed above are considerably stronger than $\mathsf{P}\neq\mathsf{NP}$ (although possibly still weaker than ETH\footnote{There is, in fact, a subclass of $\mathsf{W}[1]$, called $\mathsf{M}[1]$, of which it is known that $\mathsf{FPT}\neq \mathsf{M}[1]$ is equivalent to $\textrm{ETH}$ (see \cite{df2}). The similarities between $\mathsf{M}[1]$ and $\mathsf{W}[1]$ can be seen as a further indication that the conjecture $\mathsf{FPT}\neq \mathsf{W}[1]$ is nearly as strong as ETH, but, evidently, both $\mathsf{FPT}\neq \mathsf{M}[1]$ and $\mathsf{M}[1]\neq\mathsf{W}[1]$ are wide open conjectures.}). First, Downey and Fellows \cite{df92} construct an oracle relative to which $\mathsf{P}$ and $\mathsf{NP}$ differ while $\mathsf{W[P]}$ collapses to $\mathsf{FPT}$, so we know that any proof of the implication $\mathsf{P}\neq\mathsf{NP}\Rightarrow\mathsf{FPT}\neq\mathsf{W[P]}$ can not be as simple as the proof of the converse implication sketched above. More importantly, $\mathsf{FPT}\neq\mathsf{W[P]}$ can be related much more precisely to other classical complexity conjectures.

How strong the assumption $\mathsf{FPT}\neq\mathsf{W[P]}$ is, can be elegantly expressed in terms of \emph{limited nondeterminism}. If $f$ is a poly-time-computable function, denote by $\mathsf{NP}[f(n)]$ the class of problems that can be solved by a nondeterministic TM in polynomial-time by using at most $O(f(n))$ bits of nondeterminism ($n$ denotes the size of the input). Note that $\mathsf{NP}[\log n]=\mathsf{P}$, since a deterministic TM can cycle through all possible certificates of length $O(\log n)$ in polynomial-time. A remarkable theorem of Cai, Chen, Downey, and Fellows \cite{caichen95} states that $\mathsf{FPT}\neq\mathsf{W[P]}$ holds \emph{if and only if} for \emph{every} poly-time-computable non-decreasing unbounded function $h$, we have that $\mathsf{P}\neq\mathsf{NP}[h(n)\log n]$ (see \cite[Theorem 3.29]{fg} for a proof of the theorem in this form). The class of functions referred to in this theorem contains functions with very slow growth, such as the iterated logarithm function, $\log^\ast$. In fact, there is no poly-time-computable non-decreasing unbounded function that has the slowest growth, because if some function $h$ satisfies these conditions, then so does $\log^\ast h$. It is not even intuitively clear whether $\mathsf{P}$ is different from $\mathsf{NP}$ when the amount of allowed nondeterminism is arbitrarily close to trivial. At the very least, the fact that an infinite number of increasingly strong separations must hold in order for $\mathsf{W[P]}$ to not collapse to $\mathsf{FPT}$, suggests that separating these two classes is much farther out of our reach than a separation of $\mathsf{P}$ and $\mathsf{NP}$.

The evidence we have seen so far indicates that proving a separation of $\mathsf{W}[1]$ and $\mathsf{FPT}$ may be harder than proving $\mathsf{P}\neq\mathsf{NP}$. But assuming that both conjectures hold, it is meaningful to ask whether the internal structure of $\mathsf{W}[1]$ resembles that of $\mathsf{NP}$, and there is indeed some positive evidence in this direction. For example, a parameterized version of Cook's Theorem connects Boolean circuit satisfiability to $\mathsf{W}[1]$-completeness (see \cite{df1}), a parameterized version of Ladner's Theorem states that if $\mathsf{FPT}\neq\mathsf{W}[1]$, then there is an infinite hierarchy of problems with different complexities within $\mathsf{W}[1]$ (see \cite{df1}), and the machine-based characterizations of this class, due to Chen, Flum, and Grohe \cite{cfg3}, establish that $\mathsf{W}[1]$ can indeed be defined in terms of nondeterministic computing machines. Nevertheless, there are also previously unexplored ways in which $\mathsf{W}[1]$ may not behave the same way as $\mathsf{NP}$.

Our main goal in this work is to provide further evidence that the classes $\mathsf{FPT}$ and $\mathsf{W}[1]$ are close not only in the sense of being difficult to separate, but also in the sense that the relationship between the two differs from that of $\mathsf{P}$ and $\mathsf{NP}$, in a way that indicates that $\mathsf{FPT}$ is larger within $\mathsf{W}[1]$ than $\mathsf{P}$ is within $\mathsf{NP}$ (assuming the latter pair does not collapse). These results contrast with those in \cite{rcb17}, where we showed how certain theorems about $\mathsf{FPT}$ and the levels of the $\mathsf{A}$-Hierarchy can be proved in the same way as for their classical counterparts.

\subsection{Summary of our results}
\textbf{The difficulty of separating $\mathsf{W[P]}$ from $\mathsf{FPT}$.} Assuming that we could prove a separation of the form $\mathsf{P}\neq\mathsf{NP}[h(n)\log n]$ for a particular, slow-growing function $h$, how much progress would we have made towards proving the separation where $h(n)$ is replaced by $\log h(n)$? Intuitively, the difficulty of proving non-equality should increase when a function with a slower growth is chosen. On the other hand, if $\mathsf{FPT}\neq\mathsf{W[P]}$ holds, then all such classical separations hold as well (by the above-mentioned theorem of Cai \emph{et al.} \cite{caichen95}), and therefore any one of them implies the others. It is not clear, however, whether a proof of $\mathsf{FPT}\neq\mathsf{W[P]}$ with $\mathsf{P}\neq\mathsf{NP}[h(n)\log n]$ as a hypothesis would be significantly simpler than a proof from scratch. We show that this is unlikely to be the case, by proving (Theorem \ref{wp_fpt_thm}) that for any poly-time-computable non-decreasing unbounded function $h$, there exists a computable oracle $O_h$ such that:
\[
\mathsf{P}^{O_h}\neq\mathsf{NP}[h(n)\log n]^{O_h}\textrm{, but }\mathsf{W[P]}^{O_h}=\mathsf{FPT}^{O_h}.
\]
Theorem \ref{wp_fpt_thm} is an improvement over the above-mentioned oracle construction of Downey and Fellows \cite{df92}\footnote{Actually, Downey and Fellows \cite{df92} use a different computational model to define and relativize $\mathsf{W[P]}$, so the two results, although in the same spirit, may not be directly comparable at a technical level.}. It is weak as a barrier result, since the relativization barrier has been repeatedly overcome in the last three decades, but nevertheless the theorem succinctly expresses how much harder the conjecture $\mathsf{FPT}\neq\mathsf{W[P]}$ is compared to classical questions regarding nondeterministic vs.\ deterministic computation: No matter how small the amount of nondeterminism that provably yields a class strictly containing $\mathsf{P}$, we will always be a non-trivial proof step away from separating $\mathsf{W[P]}$ (or $\mathsf{W}[1]$) from $\mathsf{FPT}$.\\

\noindent\textbf{The structure of $\mathsf{W}[1]$ and its relation to $\mathsf{FPT}$.}
The class $\mathsf{A}[1]$(=$\mathsf{W}[1]$)\footnote{$\mathsf{W}[1]$ and $\mathsf{A}[1]$ coincide as complexity classes, but in \cite{cfg3}, Chen, Flum, and Grohe give two machine-based characterizations, one which can be generalized to get the levels of the $\mathsf{W}$-Hierarchy, and one which generalizes to the levels of the $\mathsf{A}$-Hierarchy. The machine model for $\mathsf{A}[1]$ is easier to handle when working with oracles, so we typically use this model when relativizing this class, and write ``$\mathsf{A}[1]$'' to emphasize this fact. However, oracle $\mathsf{W}[1]$-machines can also be defined so that our theorems hold for this model as well (see Section \ref{wp_fpt_sec}).} can be characterized in terms of random access machines that perform \emph{tail-nondeterministic} computations \cite{cfg3}. Such computations consist of two phases: 1.\ a (deterministic) $\mathsf{FPT}$-computation; 2.\ a short nondeterministic computation that can use any data computed in phase 1. Tail-nondeterministic machines that perform only the second phase of the computation (without a longer deterministic computation preceding it), can not solve every problem in $\mathsf{FPT}$, but, paradoxically, they can solve many problems that are complete for $\mathsf{A}[1]$ (we give an example in Section \ref{fpt_w1_sec}). As we will see, this simple observation has important consequences for the structure of this class.

A first consequence is that giving an $\mathsf{A}[1]$-machine even very restricted oracle access to a tractable ($\mathsf{FPT}$) problem, may increase its computational power, because then the use of nondeterminism can be combined with the ability to solve instances of an $\mathsf{FPT}$-problem via the oracle. Thus, $\mathsf{FPT}$-computations appear to constitute a non-trivial computational resource for $\mathsf{A}[1]$ (unlike $\mathsf{P}$-computations for $\mathsf{NP}$). Somewhat suprisingly, we can actually identify the complexity class resulting from endowing $\mathsf{A}[1]$ with $\mathsf{FPT}$-oracles, if a suitable, highly restricted type of oracle access is used. We have (Theorem \ref{fpt_low_thm}, Corollary \ref{fpt_low_cor}) that:
\[
\mathsf{A}[1]^{\mathsf{FPT}}=\mathsf{W[P]}\textrm{\ \ and\ \ }\forall t\geq 1:\mathsf{W}[t]^{\mathsf{FPT}}=\mathsf{W[P]},
\]
where we used the common notation $\mathcal{C}_1^{\mathcal{C}_2}:=\bigcup_{Q\in\mathcal{C}_2}\mathcal{C}_{1}^{Q}$. This means that either $\mathsf{W[P]}=\mathsf{W[1]}$, in which case $\mathsf{W[P]}$ is smaller than generally believed, or $\mathsf{FPT}$ is larger within $\mathsf{W}[1]$ than $\mathsf{P}$ is within $\mathsf{NP}$.

Putting the known and new facts together, Theorem \ref{wp_fpt_thm} and the result of Cai \emph{et.\ al.}\ \cite{caichen95} mentioned in the introduction indicate that $\mathsf{W[P]}$ is likely to be closer to $\mathsf{FPT}$ than any class $\mathsf{NP}[h(n)\log n]$ is to $\mathsf{P}$ (see Figure \ref{fig1}). The case for this figure being accurate is further strengthened by Theorem \ref{fpt_low_thm} and Corollary \ref{fpt_low_cor}, which exhibit another way in which at least two of the classes $\mathsf{FPT}$, $\mathsf{W}[1]$, and $\mathsf{W[P]}$ are close.

\begin{figure}[!ht]
  \begin{center}
    \begin{tikzpicture}[scale=0.95]
      \draw[thick] (0, 0) circle (1) node {$\mathsf{P}$};
      \draw[thick] (0, 1) ellipse (1.3 and 2);
      \node at (0, 2)  {$\mathsf{NP}[(\log n)^2]$};
      \draw[thick] (0, 4) ellipse (2.0 and 5);
      \node at (0, 8.5) {$\mathsf{NP}$};
      \draw[thick] (5, 0) circle (1) node {$\mathsf{FPT}$};
      \draw[thick] (5, 4) ellipse (2.0 and 5);
      \node at (5, 8.5) {$\mathsf{para\textrm{-}NP}$};
      \draw[thick] (5, 0.025) circle (1.025);
      \draw[dashed] (4.6, 0.7) rectangle (5.5, 1.3);
      \draw[very thin,dashed] (5, 1.3) -- (8, 5.5);
      \draw[very thin,dashed] (5.5, 1) -- (10, 4);
      \begin{scope}
        \draw[dashed] (8, 4) rectangle (12, 7);
        \clip (8, 4) rectangle (12, 7);
        \draw (10, -5) circle (9.4);
        \draw (10, -5.2) circle (9.8);
        \draw (10, -5.4) circle (10.2);
        \draw (10, -7.4) circle (14.2);
        \node at (10, 4.2) {$\mathsf{FPT}$};
        \node at (10, 6.5) {$\mathsf{W[P]}$};
        \node[thick] at (10, 5.2) {$\vdots$};
      \end{scope}
      \draw[thin, dotted] (11, 4.45) -- (11, 3.75);
      \node at (11, 3.5) {$\mathsf{W}[1]$};
      \draw[thin, dotted] (11.75, 4.6) -- (11.75, 3.35);
      \node at (11.75, 3.1) {$\mathsf{W}[2]$};
    \end{tikzpicture}
    \caption{The mutual closeness of the parameterized complexity classes, compared to that of their classical analogues, as suggested by \cite{caichen95, df92}, Theorems \ref{wp_fpt_thm} and \ref{fpt_low_thm}, and Corollary \ref{fpt_low_cor}. Regardless of which class $\mathsf{NP}[h(n)\log n]$ we choose to represent between $\mathsf{P}$ and $\mathsf{NP}$ on the left side (whether it is $\mathsf{NP}[(\log n)^2]$  as in the picture, $\mathsf{NP}[\log^\ast n \log n]$, or something even smaller), it will be much larger compared to $\mathsf{P}$ than $\mathsf{W[P]}$ is compared to $\mathsf{FPT}$.}\label{fig1}
  \end{center}
\end{figure}
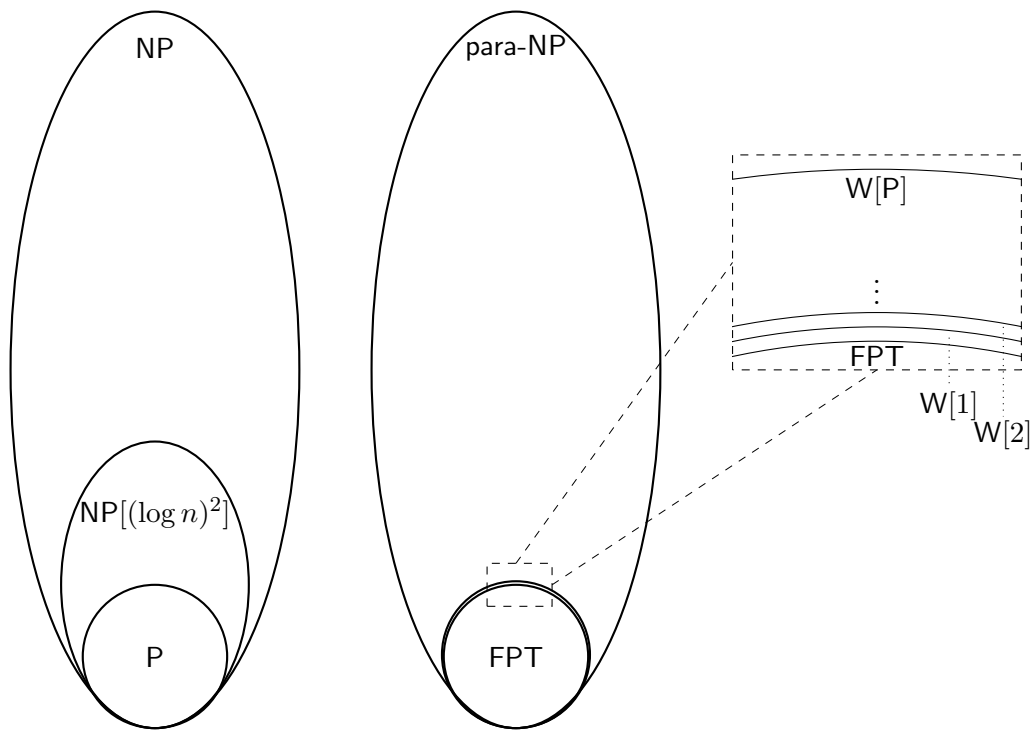

Theorem \ref{fpt_low_thm} and the observation preceding it also have consequences for the $\mathsf{A}$-Hierarchy, which is a parameterized analogue of $\mathsf{PH}$. Although they share some essential properties \cite{cfg3,rcb17}, a corollary of Theorem \ref{fpt_low_thm} is that, unless some unlikely inclusions between complexity classes occur, the two hierarchies have structural differences that indicate that consecutive levels of the $\mathsf{A}$-Hierarchy are closer to each other than the corresponding levels of $\mathsf{PH}$ (see Section \ref{fpt_w1_sec}). Conversely, using a similar idea as in the proof of Theorem \ref{fpt_low_thm}, we can show (Theorem \ref{ah_collapse_thm}) that if $\mathsf{W[P]}$ were to collapse to $\mathsf{W}[1]$ in a specific way, we would get a \emph{downward separation} theorem for the $\mathsf{A}$-Hierarchy (i.e., that if two levels collapse, the entire hierarchy collapses to the smaller of the two). Proving such a theorem for the $\mathsf{A}$-Hierarchy has been a long-standing open problem in parameterized complexity theory, and although our theorem falls short of this goal (since it requires an unlikely collapse to occur), it marks the first progress on this front in over a decade (since \cite{cfg3}).\\

\noindent\textbf{Level-by-level relativized separations of the $\mathsf{W}$- and the $\mathsf{A}$-Hierarchy.} We also give some evidence that certain collapses do not occur. The only relations that are known to hold between the classes $\mathsf{W}[t]$ and $\mathsf{A}[t]$ are that $\mathsf{W}[1]=\mathsf{A}[1]$ and that $\mathsf{W}[t]\subseteq\mathsf{A}[t]$ for $t\geq 2$. We show that in a relativized setting, the known inclusions can be made strict and some unexpected inclusions can be ruled out.

Separations of complexity classes relative to oracles count only as very weak evidence that the unrelativized versions of the classes are distinct, due to the fact that such oracles can in some cases be constructed even when two classes coincide (the most famous example being $\mathsf{IP}=\mathsf{PSPACE}$ \cite{sham} -- see \cite{fortsip} for an oracle separating the two). Nevertheless, there are a few reasons why level-by-level relativized separations for the $\mathsf{W}$- and the $\mathsf{A}$-Hierarchy are interesting: First, since it is generally assumed that these hierarchies are proper and distinct, we should expect to have at least this weak form of evidence supporting the assumption. Second, we have seen a number of results which suggest that the levels of these hierarchies are in various ways close to each other, so proving even relativized separations between them may be non-trivial. Finally, relativization in the parameterized setting is still mostly unexplored, and although the proofs of the following theorems rely on standard diagonalization arguments, the details of the machine models and how they are allowed to access oracles require special care in order to make the arguments work.

In Section \ref{wh_ah_sec} we show (Theorem \ref{wh_ah_sep_thm}, Corollary \ref{wh_ah_sep_cor}) that there exists a computable parameterized oracle $O$ such that
\[
\forall t\geq 2: \mathsf{A}[t]^O\not\subset\mathsf{W}[t]^O.
\]
Note that this is a \emph{single} oracle relative to which all inclusions are simultaneously made strict. Also note that, although we use machine-based characterizations of classes $\mathsf{A}[t]$ and $\mathsf{W}[t]$ which result in distinct characterizations of the class $\mathsf{A}[1]=\mathsf{W}[1]$, fortunately, this oracle does not appear to separate $\mathsf{A}[1]$ from $\mathsf{W}[1]$. Such a separation would have suggested that the strict inclusions are mere artifacts of the machine models used.

Finally, we give evidence which suggests that the $\mathsf{W}$-Hierarchy is not contained within any finite level of the $\mathsf{A}$-Hierarchy (Theorem \ref{wh_ah_sep_thm_2}): For all $t\geq 1$, there exists a computable parameterized oracle $O_t$ such that
\[
\mathsf{W}[t+1]^{O_t}\not\subset \mathsf{A}[t]^{O_t}.
\]
Since it holds that $\mathsf{W}[t]^{O_t}\subseteq\mathsf{A}[t]^{O_t}$, each oracle $O_t$ separates two consecutive levels of the $\mathsf{W}$-Hierarchy. This answers a question of Downey and Fellows \cite{df92}, although we do not have a single oracle that simultaneously separates the entire hierarchy.

\section{Preliminaries}
We assume familiarity with standard facts and notations from both classical and parameterized complexity theory, and refer to \cite{ab} and to \cite{fg} for the necessary background in the respective branches. Since the characterizations of various parameterized complexity classes in terms of computing machines \cite{cfg3} are less well known, we give a brief overview of the main definitions.

Many parameterized complexity classes can only be naturally characterized in terms of \emph{random access machines (RAMs)}, which can store entire integers in each of their registers, perform the operations addition, subtraction, and division by 2 on integers in unit time, and can access any part of their memory in constant time (see \cite{papadim} or the introduction of \cite{cfg3}). The input of a RAM can be a sequence of non-negative integers, and we allow the instances of problems to be encoded in this way whenever we are working only with RAMs (as opposed to TMs). Since the size of a sequence of non-negative integers is calculated as the sum of the length of the binary representations of the individual numbers, RAMs have no significant computational advantage over TMs \cite[Theorem 2.5]{papadim}. However, this encoding does make a difference when considering oracle RAMs, because the query instances will also be encoded in this fashion.

We give two examples of definitions of complexity classes in terms of RAMs. It is not difficult to see that these are equivalent to the standard (TM-based) definitions (see \cite{fg}). Note that we use the Downey-Fellows definition of parameterized problems \cite{df1}, where the parameter value, encoded in unary, is given together with the input.

\begin{defin}
Let $Q$ be a parameterized problem. We say that $Q\in\mathsf{FPT}$ if and only if there exists a RAM $M$, a computable function $f$, and a constant $c\geq 0$, such that for every input $(x,k)$ with $x\in\mathbb{N}^\ast$ and $k\geq 0$, $M$ runs in time $f(k)(|x|+k)^c$ and accepts if $(x,k)\in Q$, otherwise it rejects. The class $\mathsf{para\textrm{-}NP}$ is defined similarly, except with RAMs which can nondeterministically guess, in unit time, positive integers of size upper-bounded by the $f(k)(|x|+k)^c$ (the bound on the running time).
\end{defin}

We also collect several useful definitions and notations in the following:
\begin{defin}
Let $\mathcal{C}_1,\mathcal{C}_2$ be complexity classes, where $\mathcal{C}_1$ is defined in terms of computing machines that can be given access to an oracle, and let $P_0,P_1\subseteq\{0,1\}^\ast$ be classical languages. We define:
$\mathcal{C}_1^{\mathcal{C}_2}:=\bigcup_{P\in\mathcal{C}_2}
\mathcal{C}_1^{P}$ and $P_0\oplus P_1:=\{0x\mid x\in P_0\}\cup\{1x\mid x\in P_1\}$. We say that $P_0$ is \emph{low for $\mathcal{C}_1$} if $\mathcal{C}_1^{P_0}=\mathcal{C}_1$, and we say that $\mathcal{C}_2$ is \emph{low for $\mathcal{C}_1$} if $\mathcal{C}_1^{\mathcal{C}_2}=\mathcal{C}_1$.
\end{defin}

\subsection{The $\mathsf{A}$-Hierarchy and the $\mathsf{W}$-Hierarchy}\label{prelim_sec1}
The following classes are defined in terms of \emph{alternating random access machines (ARAMs)}, which are RAMs that can nondeterministically guess, in unit time, integers of size bounded by the running time of the machine on a given input, either in the existential or the universal mode (see \cite{cfg3}).
\begin{defin}[\cite{cfg3}]\label{At_def}
For each $t\geq 1$, let $\mathsf{A}[t]$ be the class of parameterized problems that are solved by some ARAM $A$ which, for some computable functions $f$ and $h$, and a constant $c\geq 0$, satisfies the following conditions on every input $(x,k)$:
\begin{itemize}
\item[1.] $A$ runs in time at most $f(k)(|x|+k)^c$;
\item[2.] throughout the computation, the values in $A$'s registers do not exceed $f(k)(|x|+k)^c$;
\item[3.] all nondeterministic guesses are made during the last $h(k)$ steps of the computation;
\item[4.] the first nondeterministic guess is existential and the machine alternates at most $t-1$ times between existential and universal guesses.
\end{itemize}
The class $\mathsf{co\textrm{-}A}[1]$ is defined in terms of ARAMs which satisfy conditions 1--3, but only make universal nondeterministic guesses (one can verify, just as in the classical setting, that a problem is in $\mathsf{co\textrm{-}A}[1]$ if and only if it is the complement of a problem in $\mathsf{A}[1]$). ARAMs satisfying conditions 1 and 2 are called \emph{parameter-bounded} in \cite{rcb17}, those satisfying conditions 3 and 4 are called, respectively, \emph{tail-nondeterministic} and \emph{$t$-alternating} \cite{cfg3}.
\end{defin}

The classes $\mathsf{W}[t]$ ($t\geq 1$) can be defined in terms of $\mathsf{A}[t]$-machines (parameter-bounded tail-nondeterministic $t$-alternating ARAMs) that are further restricted so that: 1. Every block of nondeterministic guess instructions of the same kind, except the first one, is made up of at most $c'$ guess instructions, where $c'$ is a constant that is independent of the input. 2. All nondeterministically guessed integers are placed in a special set of \emph{guess registers}, which can not be read from directly, and can only be accessed via special instructions that use the guessed values as indices for accessing standard registers. We will not need further details regarding these machines, and therefore refer the reader to \cite{cfg3} or \cite{bussislam06} for more complete definitions. We will, however, define oracle $\mathsf{W}[t]$-machines (Definition \ref{wt_defin}).

\begin{defin}[\cite{rcb17}]\label{oracle_def1}
An \emph{oracle (A)RAM} is a machine with an additional set of registers called \emph{oracle registers}, instructions that allow the machine to copy values from its standard registers to the oracle registers, as well as a QUERY instruction, the execution of which results in one of the values 1 or 0 being placed into the first standard register of the machine, depending on whether the instance encoded in the oracle registers at that time constitute a `yes'- or a 'no'-instance of a problem for which the machine is said to have an oracle.

An oracle (A)RAM has \emph{balanced oracle access} to a parameterized oracle, if there is a computable function $g$ such that on every input $(x,k)$, the machine queries the oracle only with instances whose parameter value is $\leq g(k)$ (in other words, the parameter values of the instances for which the oracle is called should be upper-bounded by some function of $k$, but may not depend on $n$, even though the machine may have time to construct such a query instance). An oracle (A)RAM has \emph{tail-restricted oracle access}, if its access to the oracle is balanced and, furthermore, there is a computable function $h$ such that the machine makes oracle queries only within the last $h(k)$ steps of the computation on input $(x,k)$. Note that tail-restricted access is also balanced.

For a parameterized complexity class $\mathcal{C}$ that is defined in terms of (A)RAMs, we write $\mathcal{C}(O)$ if $\mathcal{C}$ has unrestricted access to the oracle $O$, $\mathcal{C}(O)_{bal}$ if it has balanced access, and $\mathcal{C}(O)_{tail}$ if it has tail-restricted access. If $\mathcal{C}$ is defined in terms of tail-nondeterministic ARAMs, we also write $\mathcal{C}^O$ instead of $\mathcal{C}(O)_{tail}$ (so $\mathsf{A}[1]^O$ means $\mathsf{A}[1]$ with tail-restricted access to $O$). Note that for $\mathcal{C}(O)$, the oracle can be either classical or parameterized, but for balanced or more restricted oracle access, it must be parameterized.
\end{defin}

We define oracle access for $\mathsf{W}[t]$-machines in Section \ref{wp_fpt_sec}.
\subsection{$\mathsf{W[P]}$ and the $\mathsf{W[P]}$-Hierarchy}
We define $\mathsf{W[P]}$ both in terms of TMs and in terms of RAMs, and use both definitions at different points in the paper.
\begin{defin}[\cite{cfg3}]\label{wp_def}
Let $Q$ be a parameterized problem. We say that $Q\in\mathsf{W[P]}$ if and only if there exists a nondeterministic TM $M$, computable functions $f$ and $h$, and a constant $c\geq 0$, such that for any input $(x,k)$ with $x\in\{0,1,\#\}^\ast$ and $k\geq 0$, $M$ runs in time $f(k)(|x|+k)^c$, uses at most $h(k)\lceil\log(|x|+k)\rceil$ nondeterministic bits, and accepts if and only if $(x,k)\in Q$. 
\end{defin}

The following problem is complete for $\mathsf{W[P]}$ under fpt-reductions.
\begin{center}
\fbox{
\begin{minipage}{12.5cm}
$p\textsc{-WSatCircuit}$\\
\begin{tabular}{ r l }
Input: & \parbox[t]{9cm}{A Boolean circuit $C$ with $n$ input bits, $k\in\mathbb{N}$.}\\
Parameter: & $k$\\
Problem: & Decide whether $C$ has a satisfying assignment of weight $k$.
\end{tabular}
\end{minipage}}
\end{center}

The class $\mathsf{W[P]}$ can also be defined in terms of RAMs \cite{cfg3}. One can also define a hierarchy that is similar to $\mathsf{PH}$, except in terms of alternating nondeterminism that matches the nondeterminism of $\mathsf{W[P]}$.

\begin{defin}[\cite{rcb17}]\label{wph_def}
For each $t\geq 1$, let $\mathsf{\Sigma}_t^{[P]}$ be the class of parameterized problems that are solved by some ARAM $A$ which, for some computable functions $f$ and $h$, and a constant $c\geq 0$, satisfies, on every input $(x,k)$, conditions 1, 2, and 4 from Definition \ref{At_def}, as well as:
\begin{itemize}
\item[3'.] $A$ nondeterministically guesses at most $h(k)$ numbers throughout the computation.
\end{itemize}
We denote the class $\bigcup_{t=1}^{\infty}\mathsf{\Sigma}_t^{[P]}$ by $\mathsf{W[P]H}$, the \emph{$\mathsf{W[P]}$-Hierarchy}.
\end{defin}

It is not difficult to see that $\mathsf{W[P]}=\mathsf{\Sigma}^{[P]}_1$. Note that we will use the term ``$\mathsf{W[P]}$-machine'' to designate both the TMs from Definition \ref{wp_def} as well as the $\mathsf{\Sigma}^{[P]}_1$-machines from Definition \ref{wph_def} (which are RAMs), but it should be clear from the context which type of machine is meant.

For each $t\geq 1$, the following generalizations of the problem $p\textsc{-WSatCircuit}$ can easily be seen to be, respectively, complete problems for $\mathsf{\Sigma}^{[P]}_t$.

\begin{center}
\fbox{
\begin{minipage}{12cm}
$p\textsc{-AWSatCircuit}_t$\\
\begin{tabular}{ r l }
Input: & \parbox[t]{9cm}{A Boolean circuit $C$ with $n$ input bits, $k\in\mathbb{N}$, and a partition of the input variables of $C$ into $t$ sets $I_1,\ldots,I_t$.}\\
Parameter: & $k$\\
Problem: & \parbox[t]{9cm}{Decide whether there exists a set $J_1\subseteq I_1$ of size $k$ such that for all subsets $J_2\subseteq I_2$ of size $k$ there exists ... such that setting precisely the variables in $J_1\cup\ldots\cup J_t$ to `true' results in a satisfying assignment of $C$.}
\end{tabular}
\end{minipage}}
\end{center}

As in the case of the Polynomial Hierarchy, and unlike the case of the $\mathsf{A}$-Hierarchy, it is known that the collapse of levels of the $\mathsf{W[P]}$-Hierarchy would propagate upward:
\begin{fact}[Corollary 17 of \cite{rcb17}]\label{wph_collapse_fact}
If for any $t\geq 1$, $\mathsf{\Sigma}^{[P]}_{t+1}=\mathsf{\Sigma}^{[P]}_{t}$, then $\mathsf{W[P]H}=\mathsf{\Sigma}^{[P]}_{t}$.
\end{fact}

\begin{defin}[\cite{rcb17}]\label{oracle_def2}
An oracle ARAM has \emph{parameter-bounded oracle access} to a parameterized oracle, if its access to the oracle is balanced and, furthermore, there is a computable function $g$ such that on every input $(x,k)$, the machine makes at most $g(k)$ oracle queries.

If $\mathcal{C}$ is a class that is defined in terms of ARAMs, we write $\mathcal{C}(O)_{para}$ to denote that $\mathcal{C}$ has parameter-bounded access to the oracle $O$. If $\mathcal{C}=\mathsf{\Sigma}_t^{[P]}$, for some $t\geq 1$, we may also write $\mathcal{C}^O$ to mean $\mathcal{C}(O)_{para}$ (so $\mathsf{W[P]}^O=\mathsf{W[P]}(O)_{para}$).
\end{defin}

\section{The difficulty of separating $\mathsf{W[P]}$ from $\mathsf{FPT}$}\label{wp_fpt_sec}
In this section we show that there is likely no shortcut to proving $\mathsf{FPT}\neq\mathsf{W[P]}$ via any finite number of separations of the form $\mathsf{P}\neq\mathsf{NP}[h(n)\log n]$. For the sake of readability, the proofs of the theorems in this section have been moved to the appendix.

To prove the theorem, we need to construct an oracle relative to which two conditions hold simultaneously: the collapse of one pair of complexity classes and the separation of another. One approach to achieving this is to construct the oracle in stages, and to work towards one goal in the odd-numbered stages and towards the other in the even-numbered ones, while ensuring that the two constructions do not interfere with each other (see \cite[Theorem 5.1]{bookwilxu} for one example of an application of this technique). However, this approach does not always work, and in this case it fails because one pair of classes is parameterized (specifically, it is not possible to computably list all $\mathsf{FPT}$- or $\mathsf{W[P]}$-machines, but this appears to be necessary in this type of staged construction). To overcome this obstacle, we use an idea of Allender \cite{allender91}, who constructs an oracle with two parts: the first part is designed so as to ensure that one pair of classes collapses \emph{regardless of what the second part of the oracle is}; the second part can then be freely used in a diagonalization argument to separate the remaining pair of classes.
\begin{theor}\label{wp_fpt_thm}
For every polynomial-time-computable non-decreasing unbounded function $h$, there exists a computable oracle $B$ such that
\begin{align*}
\mathsf{P}^B\neq \mathsf{NP}[h(n)\log n]^{B}\textrm{,\ \ but\ \ }\mathsf{W[P]}(B)=\mathsf{FPT}(B).
\end{align*}
\end{theor}
For this result we have used unrestricted oracle access to relativize the parameterized complexity classes, rather than the parameter-bounded type that we argued is natural for $\mathsf{W[P]}$ \cite{rcb17}. This is because restricting the oracle access to being balanced (or more) makes it possible to collapse even $\mathsf{para\textrm{-}NP}$ to $\mathsf{FPT}$, with the classical separation unchanged. Thus we would get an oracle relative to which $\mathsf{NP}$ and $\mathsf{P}$ differ while $\mathsf{para\textrm{-}NP}$ and $\mathsf{FPT}$ coincide, which is clearly an artifact of the restrictions placed on the oracle access, since we know that, unrelativized, a collapse of $\mathsf{para\textrm{-}NP}$ to $\mathsf{FPT}$ is equivalent to a collapse of $\mathsf{NP}$ to $\mathsf{P}$ (see \cite{fg}).

It seems reasonable to expect that if $\mathsf{W[P]}$ collapses to $\mathsf{FPT}$ relative to some oracle, then so should any class $\mathsf{W}[t]$. But to show that this is indeed the case, we first need to define oracle $\mathsf{W}[t]$-machines. Recall that a $\mathsf{W}[t]$-machine is similar to an $\mathsf{A}[t]$-machine, but the numbers it guesses nondeterministically are placed in a special set of \emph{guess registers}, to which the machine has only limited access \cite{cfg3} (see also \cite{bussislam06}). Naturally, an oracle $\mathsf{W}[t]$-machine should then have three sets of registers: the standard registers, guess registers (which the machine can not read from directly), and oracle registers. For such machines, the usual way to read from or write to the oracle registers is very limiting, because the machines' nondeterminism would only weakly be able to influence the query instances. For example, nondeterministically guessed numbers could not be written to the oracle registers. The interaction between the nondeterminism of such machines and their ability to form query instances can be strengthened without allowing the $\mathsf{W}[t]$-specific restrictions to be circumvented. We achieve this by making the oracle registers \emph{write-only}, and adding instructions that allow the machine to copy values from the guess registers to the oracle registers and to use numbers from the guess registers to address oracle registers. In this way, the machine can still not read the guessed numbers directly or use them in arithmetic computations, but can nevertheless use them for oracle queries. In many cases, this allows oracle $\mathsf{W}[t]$-machines to match $\mathsf{A}[t]$-machines in the way the oracle is used. 

\begin{defin}\label{wt_defin}
An \emph{oracle $\mathsf{W}[t]$-machine} is a $\mathsf{W}[t]$-machine that, in addition to the standard registers $r_0,r_1,\ldots$, and guess registers $g_0,g_1,\ldots$ (to which the machine has only restricted access), also possesses a set of oracle registers $o_0,o_1,\ldots$. The contents of oracle registers are never read from and are only affected by the following new instructions:
\begin{itemize}
\item[] SO\_MOVE - copy the contents of standard register $r_0$ to oracle register $o_{r_1}$;
\item[] GO\_MOVE - copy the contents of guess register $g_{r_0}$ to oracle register $o_{r_1}$;
\item[] ADDR\_GO\_MOVE - copy the contents of register $r_0$ to $o_{g_{r_1}}$;
\item[] OO\_MOVE - copy the contents of $o_{r_0}$ to $o_{r_1}$.
\end{itemize}
Additionally, the machine has a QUERY instruction that places either the value $0$ or $1$ into $r_0$, depending on whether the contents of the oracle registers at the time when the instruction is executed represent a `no'- or a 'yes'-instance of the problem to which the machine has oracle access.
\end{defin}

Note that for such machines, it again makes sense to speak of unrestricted, balanced, parameter-bounded, or tail-restricted oracle access. With the above definition in mind, we can now prove Corollary \ref{wp_fpt_cor}, where oracle access is unrestricted.

\begin{corol}\label{wp_fpt_cor}
For any function $h$ as in Theorem \ref{wp_fpt_thm}, and the corresponding oracle $B$, we have that $\mathsf{FPT}(B)=\mathsf{W}[1](B)=\mathsf{A}[1](B)=\mathsf{W}[2](B)=\mathsf{A}[2](B)=\ldots=\mathsf{W[P]}(B)$.
\end{corol}

\section{The structure of $\mathsf{W}[1]$ and its relation to $\mathsf{FPT}$}\label{fpt_w1_sec}
Under the assumption that $\mathsf{P}\neq\mathsf{NP}$, it is meaningful to ask whether the relation between the two classes is the same as the one between $\mathsf{FPT}$ and $\mathsf{W}[1]$. So far, we have seen evidence that the two parameterized classes are closer to each other in the sense that proving a separation between them is more difficult than proving $\mathsf{P}\neq\mathsf{NP}$. In this section we look at other ways in which $\mathsf{W}[1]$ is closer to $\mathsf{FPT}$ than $\mathsf{NP}$ is to $\mathsf{P}$.

In this section, the definitions of classes in terms of RAMs are used, instances of problems and oracles query instances are encoded as integer sequences, and oracles are parameterized.\\

\noindent\emph{Is $\mathsf{FPT}$ low for $\mathsf{W}[1]$?}

Given that $\mathsf{FPT}$ is the class of tractable problems in parameterized complexity, and that $\mathsf{P}$-oracles add no computational power to $\mathsf{NP}$ (or to any class $\mathsf{NP}[h(n)\log n]$), one might expect $\mathsf{FPT}$ to also be low for $\mathsf{W}[1]$. It turns out, however, that allowing tail-nondeterministic machines to make even tail-restricted queries to an $\mathsf{FPT}$-oracle can increase their computational strength to that of $\mathsf{W[P]}$. We prove this for $\mathsf{A}[1]$ first, since this machine model is more easily relativizable.

\begin{theor}\label{fpt_low_thm}
$\mathsf{A}[1]^{\mathsf{FPT}}=\mathsf{W[P]}$. Therefore, $\mathsf{FPT}$ is low for $\mathsf{A}[1]$ if and only if $\mathsf{W[P]}=\mathsf{A}[1]$ and the $\mathsf{W}$-Hierarchy collapses to its first level.
\end{theor}
\begin{proof}
We have that $\mathsf{A}[1]^{\mathsf{FPT}}\subseteq\mathsf{A}[1](\mathsf{FPT})_{bal}\subseteq\mathsf{W[P]}(\mathsf{FPT})_{bal}=\mathsf{W[P]}$, with the final equality holding because a $\mathsf{W[P]}$-machine can replace balanced oracle calls to $\mathsf{FPT}$-problems by fpt-length computations.

To show that $\mathsf{W[P]}\subseteq \mathsf{A}[1]^{\mathsf{FPT}}$, we define the following problem:
\begin{center}
\fbox{
\begin{minipage}{12cm}
$p\textsc{-WSatCircuit-with-assignment}$\\
\begin{tabular}{ r l }
Input: & \parbox[t]{9cm}{A circuit $C$ with $n$ inputs, $k\in\mathbb{N}$, and vector $v\in\{0,1\}^n$ of weight $k$.}\\
Parameter: & $k$.\\
Problem: & Decide whether $v$ is a satisfying assignment for $C$.
\end{tabular}
\end{minipage}}
\end{center}
Since the output of a circuit can be computed in time polynomial in its size, the above problem is obviously in $\mathsf{FPT}$\footnote{In fact, this problem is clearly in $\mathsf{P}$, meaning that we can actually prove the stronger statement $\mathsf{W[P]}\subseteq\mathsf{A}[1]^{\mathsf{P}}$. However, we choose $\mathsf{FPT}$ instead of $\mathsf{P}$ because the instance with which the oracle is queried will be fpt-sized, and because it is more natural to have $\mathsf{A}[1]$-machines query a parameterized oracle, rather than a classical one.}. Any problem $Q\in\mathsf{W[P]}$ can be solved by some $\mathsf{A}[1]$-machine $A$ with tail-restricted access to $p$\textsc{-WSatCircuit-with-assignment} as an oracle: First, $A$ reduces in fpt-time the input instance $(x,k)$ to an instance $(y,k')$ of $p$\textsc{-WSatCircuit}, where $k'$ depends computably only on $k$. Let $m$ be the number of input bits of the circuit encoded in $y$. If $m<k'$, $A$ rejects, otherwise it writes $y$, $0^m$, and $1^{k'}$ to its oracle registers, thus forming a valid instance of $p$\textsc{-WSatCircuit-with-assignment}, except that the assignment vector has weight $0$. Now $A$ enters the nondeterministic phase of its computation by guessing $k'$ pairwise distinct integers $i_1,\ldots,i_{k'}\in[m]$. It then modifies the assignment vector in the oracle registers by changing the zeroes at positions $i_1,\ldots,i_{k'}$ of the vector $0^m$ to $1$, queries the oracle, and accepts if the answer is `yes', otherwise it rejects.

It is easy to see that what this machine actually does is nondeterministically guess a satisfying assignment of the $p$\textsc{-WSatCircuit}-instance, if one exists, and delegate the verification to the oracle. The trick here is that the all-zero assignment vector must be written to the oracle registers deterministically, because during the nondeterministic phase at the end of the computation there may not be enough time to do so. Then the machine only needs to change the vector at $k'$ positions to obtain an assignment with the right weight, which takes only $O(k')$ steps with random access memory.
\end{proof}
Note that the proof that a $\mathsf{W[P]}$-machine can simulate $\mathsf{A}[1]$-machines with $\mathsf{FPT}$-oracles only works if the oracle access of the $\mathsf{A}[1]$-machines is tail-restricted or at least parameter-restricted. On the other hand, the proof that $\mathsf{A}[1]^{\mathsf{FPT}}\supset\mathsf{W[P]}$ only requires tail-restricted oracle access. We regard this as further evidence (in addition to the results from \cite{rcb17}) that tail-restricted oracle access is the natural type to consider for the class $\mathsf{A}[1]$.

Since $\mathsf{A}[1]\subseteq\mathsf{W}[t]\subseteq\mathsf{W[P]}$ holds for all $t\geq 1$, and by Theorem \ref{fpt_low_thm} we have that $\mathsf{A}[1]^{\mathsf{FPT}}=\mathsf{W[P]}=\mathsf{W[P]}(\mathsf{FPT})_{tail}$, it seems reasonable to expect that $\mathsf{W}[t]^{\mathsf{FPT}}=\mathsf{W[P]}$ holds for all $t$ as well. In order to prove that this is indeed the case, we need to use oracle $\mathsf{W}[t]$-machines (Definition \ref{wt_defin}), with tail-restricted access to the oracle. Then the proof is based on a combination of ideas from the proofs of Corollary \ref{wp_fpt_cor} and Theorem \ref{fpt_low_thm}.
\begin{corol}\label{fpt_low_cor}
For every $t\geq 1$ it holds that $\mathsf{W}[t]^{\mathsf{FPT}}=\mathsf{W[P]}$.
\end{corol}
In \cite{rcb17} we showed that for every $t\geq 1$ the class $\mathsf{A}[t+1]$ can be obtained as $\mathsf{A}[1]^{O_t}$, where $O_t$ is a specific $\mathsf{A}[t]$-complete oracle, but we also observed that $\mathsf{A}[1]^{\mathsf{FPT}}$ does not appear to be a subset of $\mathsf{A}[t]$ for any $t$. Theorem \ref{fpt_low_thm} provides support for this intuition by identifying $\mathsf{A}[1]^{\mathsf{FPT}}$ as a class which is not known or believed to be a subset of any class $\mathsf{A}[t]$.
\begin{corol}\label{fpt_low_cor2}
For every $t\geq 1$ we have that if $\mathsf{A}[t+1]=\mathsf{A}[1]^{\mathsf{A}[t]}$, then $\mathsf{W[P]}\subset\mathsf{A}[t+1]$. In particular, if $\mathsf{W[P]}\not\subset\mathsf{A}[2]$, we have that $\mathsf{A}[1]^{\mathsf{A}[1]}\neq \mathsf{A}[2]$.
\end{corol}
Corollary \ref{fpt_low_cor2} shows that the above-mentioned oracle characterization of the $\mathsf{A}$-Hierarchy from \cite{rcb17} can probably not be improved significantly: Although it may be possible to obtain $\mathsf{A}[t+1]$ by providing $\mathsf{A}[1]$ with other $\mathsf{A}[t]$-complete oracles, it is unlikely that $O_t$ can be replaced by the entire class $\mathsf{A}[t]$, for the somewhat counter-intuitive reason that $\mathsf{A}[t]$ contains all \emph{tractable} problems. More importantly, Corollary \ref{fpt_low_cor2} implies, assuming $\mathsf{W[P]}\not\subset\mathsf{A}[t+1]$ and that $\mathsf{PH}$ is proper, that each class $\mathsf{A}[t+1]$ is closer to the class $\mathsf{A}[t]$ than $\mathsf{\Sigma}^P_{t+1}$ is to $\mathsf{\Sigma}^P_{t}$, in the precise sense that $\mathsf{A}[t+1]\subsetneq \mathsf{A}[1]^{\mathsf{A}[t]}$, whereas $\mathsf{\Sigma}^P_{t+1}=\mathsf{NP}^{\mathsf{\Sigma}^P_{t}}$. The use of a highly restricted type of oracle access for the parameterized classes can only make this conclusion more legitimate.\\

We also mention one failed attempt to use Theorem \ref{fpt_low_thm}: It is natural to ask whether $\mathsf{W}[1]=\mathsf{FPT}$ would imply a collapse of larger parameterized $\mathsf{NP}$-analogues to $\mathsf{FPT}$. Currently, it is not even known if $\mathsf{W}[1]=\mathsf{FPT}\Rightarrow\mathsf{W}[2]=\mathsf{FPT}$. It would appear that the identity $\mathsf{A}[1]^{\mathsf{FPT}}=\mathsf{W[P]}$ offers a way to prove such statements via the following argument: If $\mathsf{A}[1]=\mathsf{FPT}$, then, by Theorem \ref{fpt_low_thm} and the fact that $\mathsf{FPT}$ is low for itself, we should be able to conclude that $\mathsf{W[P]}=\mathsf{A}[1]^{\mathsf{FPT}}=\mathsf{FPT}(\mathsf{FPT})_{tail}=\mathsf{FPT}$. Unfortunately, this argument fails because the property of being self-low is sensitive to the machine model used to define a class. Thus, we might have that $\mathsf{A}[1]=\mathsf{FPT}$, while $\mathsf{A}[1]^{\mathsf{FPT}}\neq\mathsf{A}[1]$. We give an example of this situation occurring in the classical setting:

\begin{prop}\label{counter_example_prop}
There exists a classical complexity class with two machine characterizations, denote them respectively by $\mathcal{C}_{1}$ and $\mathcal{C}_{2}$ (thus, $\mathcal{C}_{1}=\mathcal{C}_{2}$ as complexity classes), such that $\mathcal{C}_{1}^{\mathcal{C}_{1}}=\mathcal{C}_{1}$, but $\mathcal{C}_{2}^{\mathcal{C}_{2}}\neq \mathcal{C}_{2}$. In other words, the complexity class is self-low when defined in terms of one type of oracle machines, but not when defined in terms of the other type.
\end{prop}
\begin{proof}
Let $\mathcal{C}_1$ be the class $\mathsf{IP}$, and let $\mathcal{C}_2$ be $\mathsf{PSPACE}$, relativized so that the use of the oracle tape does not count towards the machine's space usage (this model has been studied exensively; see, for example, \cite{gavtoren}). Since the verifier in an interactive proof system is polynomial-time-bounded, any queries he makes to an oracle must also be polynomial-sized. Since $\mathsf{PSPACE}=\mathsf{IP}\subseteq\mathsf{IP}^\mathsf{IP}=\mathsf{IP}^\mathsf{PSPACE}\subseteq\mathsf{PSPACE}^{\mathsf{PSPACE}}$(with polynomial-sized queries) $=\mathsf{PSPACE}$, we have that $\mathcal{C}_1^{\mathcal{C}_1}=\mathcal{C}_1$. But a $\mathcal{C}_2$-machine with an oracle for $\mathsf{PSPACE}$ can make queries that are very large, and we get that $\mathcal{C}_2^{\mathcal{C}_1}=\mathsf{EXPSPACE}\neq\mathsf{PSPACE}$.
\end{proof}

\noindent\emph{A weak downward separation theorem for the $\mathsf{A}$-Hierarchy.}

It is a long-standing open problem whether the collapse of any class $\mathsf{A}[t+1]$ to $\mathsf{A}[t]$ would cause all higher levels of the $\mathsf{A}$-Hierarchy to coincide with $\mathsf{A}[t]$ (or, equivalently, whether a separation of two classes $\mathsf{A}[t+1]$ and $\mathsf{A}[t]$ would imply that all levels below $\mathsf{A}[t]$ are distinct, whence the name ``downward separation''). Given the similarities with $\mathsf{PH}$, one might expect such a theorem to hold for the $\mathsf{A}$-Hierarchy as well. Nevertheless, the proof of the downward separation theorem for the Polynomial Hierarchy does not appear to carry over directly to the parameterized setting. So far, the best result in this direction has been a theorem of Chen \emph{et al.}\ \cite{cfg3}, who showed that $\mathsf{W[P]}=\mathsf{FPT}$ implies $\mathsf{FPT}=\mathsf{A}[1]=\mathsf{A}[2]=\ldots$. This result is already non-trivial, since $\mathsf{A}[t]$ is not known to be a subset of $\mathsf{W[P]}$ for $t>1$, and can be viewed as a parameterized version of $\mathsf{P}=\mathsf{NP}\Rightarrow\mathsf{PH}=\mathsf{P}$, except that the stronger collapse $\mathsf{W[P]}=\mathsf{FPT}$ is required instead of $\mathsf{A}[1]=\mathsf{FPT}$ (in fact, the proof of the parameterized theorem in \cite{cfg3} is adapted from the proof of the corresponding classical theorem). Previously it was not known whether assuming a weaker collapse, for example $\mathsf{W[P]}=\mathsf{A}[1]$, might also suffice to prove that $\forall t\geq 1:\mathsf{A}[t]=\mathsf{A}[t+1]\Rightarrow \mathsf{A}\textrm{-Hierarchy}=\mathsf{A}[t]$. In what follows we prove such a theorem.

Let $\mathsf{A}[1]_c$ be the class of parameterized problems $Q$ such that there exists an $\mathsf{A}[1]$-machine that solves any instance $(x,k)$ of $Q$ in a number of steps depending only on $k$. This subclass of $\mathsf{A}[1]$ contains the problems that can be solved by $\mathsf{A}[1]$-machines \emph{without} the need for a precomputation that runs in fpt-time. It is provably not closed under fpt-reductions, but contains many important $\mathsf{W}[1]$-complete problems, provided that the input is given in an appropriate format. For example, $p\textsc{-IndependentSet}\in\mathsf{A}[1]_c$, if the input graph is given in the form of an adjacency matrix, because then an $\mathsf{A}[1]$-machine can first guess $k$ vertices (recall that a nondeterministic RAM can guess an integer between $1$ and $n$ in a single step; see Section \ref{prelim_sec1}) and use its random access memory to verify in $O(k^2)$ steps that none of the edges between two guessed vertices are in the graph. One can similarly show that $p\textsc{-ShortTMAcceptance}$ and other $\mathsf{W}[1]$-complete problems are in $\mathsf{A}[1]_c$.

If $\mathsf{W[P]}$ were to collapse to $\mathsf{A}[1]$, then the $\mathsf{W[P]}$-complete problem $p\textsc{-WSatCircuit}$ would also be $\mathsf{A}[1]$-complete, and therefore it would seem reasonable to expect that it is also in $\mathsf{A}[1]_c$, given an appropriate, efficiently computable encoding of the input. Thus, $p\textsc{-WSatCircuit}\in\mathsf{A}[1]_c$ seems only slightly less likely than $\mathsf{W[P]}=\mathsf{A}[1]$ (although, strictly speaking, both $p\textsc{-WSatCircuit}\in\mathsf{A}[1]_c$ and $p\textsc{-WSatCircuit}\in\mathsf{FPT}$ (used by Chen \emph{et al.}\ \cite{cfg3}) are strictly stronger assumptions than $p\textsc{-WSatCircuit}\in\mathsf{A}[1]$, and probably mutually incomparable). Under this assumption, we can prove the following:

\begin{theor}\label{ah_collapse_thm}
Assume that $p\textsc{-WSatCircuit}\in\mathsf{A}[1]_c$, meaning that there exists an $\mathsf{A}[1]$-machine that solves any instance $(x,k)$ of $p\textsc{-WSatCircuit}$ in a number of steps depending computably on $k$ alone. Then for all $t\geq 1$ we have that $\mathsf{A}[t]=\mathsf{A}[t+1]\Rightarrow(\forall u\geq 1:\mathsf{A}[t]=\mathsf{A}[t+u])$.
\end{theor}
\begin{proof}
We show that under the first assumption in the theorem statement, we have for every $t\geq 1$ that $\mathsf{A}[t+1]=\mathsf{\Sigma}^{[P]}_{t+1}$. Since we already have a downward separation theorem for $\mathsf{W[P]H}$ (Fact \ref{wph_collapse_fact}), it follows that the desired conclusion holds for the $\mathsf{A}$-Hierarchy.

First, we have for every $t\geq 1$ that $\mathsf{\Sigma}^{[P]}_{t+1}\subseteq \mathsf{A}[t]^{p\textsc{-WSatCircuit}}$, by a similar proof as that of Theorem \ref{fpt_low_thm}: To solve a problem $Q\in\mathsf{\Sigma}^{[P]}_{t+1}$, an $\mathsf{A}[t]^{p\textsc{-WSatCircuit}}$-machine will first compute a reduction to the canonical $\mathsf{\Sigma}^{[P]}_{t+1}$-complete problem $p\textsc{-AWSatCircuit}_{t+1}$, and, if $t$ is odd, modify the resulting circuit so that its output is flipped. The machine then uses its $t$-alternating nondeterminism to guess the variables to set to 1 in the first $t$ sets of the partition of the circuit's inputs, and hardwires this partial assignment into the circuit. The result is an instance of $p\textsc{-WSatCircuit}$, which can be solved with a single query to the oracle, and the oracle $\mathsf{A}[t]$-machine now outputs the oracle's answer if $t$ is odd, otherwise it outputs the opposite answer. It is easy to verify that this solves the problem $Q$.

Finally, we outline the proof that $\mathsf{A}[t]^{p\textsc{-WSatCircuit}}\subseteq\mathsf{A[t+1]}$, under the assumption that the algorithm for $p\textsc{-WSatCircuit}$ mentioned in the theorem statement exists. This inclusion is proved in the same manner as $\mathsf{A}[1]^{p\textsc{-MC}(\Sigma_t[3])}\subseteq\mathsf{A}[t+1]$ \cite[Theorem 13]{rcb17}, which is itself a parameterized version of the proof of the well-known fact that $\mathsf{NP}^{\Sigma_t\textsc{SAT}}\subseteq \mathsf{\Sigma}_{t+1}^{P}$ (see \cite[Section 5.5]{ab}). An $\mathsf{A[t+1]}$-machine can first perform the deterministic part of the oracle $\mathsf{A}[t]$-machine's computation, and then use its $(t+1)$-alternating nondeterminism to guess the answers to the subsequent oracle queries of the simulated machine (existentially), all of its $t$-alternating nondeterministic guesses, as well as (suitably quantified) witnesses for the query instances. Oracle queries are then replaced by computations in which the guessed witnesses are used instead of nondeterministic guesses. The fact that evaluations of $p\textsc{-WSatCircuit}$-queries can be performed in this manner, is due to the assumption that this problem has a nondeterministic algorithm running in time dependent on $k$ alone.

Since $\mathsf{A}[t+1]\subseteq \mathsf{\Sigma}^{[P]}_{t+1}$ holds unconditionally, we conclude that the two classes are equal, which completes the proof.
\end{proof}

\section{Level-by-level relativized separations of the $\mathsf{W}$- and the $\mathsf{A}$-Hierarchy}\label{wh_ah_sec}
In this section we give oracle-based evidence that the main parameterized hierarchies do not collapse in unforeseen ways. We start by constructing a single oracle relative to which the inclusion of every $\mathsf{W}[t]$ in $\mathsf{A}[t]$ is strict, except for the first level. In fact, we accomplish this by proving the strongest possible relativized separation between co-nondeterminism and (existential) nondeterminism in the parameterized setting: the weakest co-nondeterministic class with tail-restricted oracle access, against the strongest nondeterministic class with unrestricted oracle access.

The proofs of the theorems in this section are based on standard diagonalization arguments that have been adapted to the parameterized setting, and can be found in the appendix.

\begin{theor}\label{wh_ah_sep_thm}
There exists a computable oracle $O$ such that $\mathsf{co\textrm{-}A}[1]^O\not\subset \mathsf{para\textrm{-}NP}(O)$.
\end{theor}

Since $\mathsf{co\textrm{-}A}[1]^O\subseteq\mathsf{A}[t]$ for all $t\geq 2$, and $\mathsf{W}[t]^{O}\subseteq\mathsf{para\textrm{-}NP}(O)$ for all $t\geq 1$, we immediately get the next corollary. Note, however, that the oracle constructed here does not appear to separate $\mathsf{A}[1]$ from $\mathsf{W}[1]$, since the separating problem in $\mathsf{co\textrm{-}A}[1]^O\setminus \mathsf{para\textrm{-}NP}(O)$ is not in $\mathsf{A}[1]^O$. Had a separation of two coinciding classes occured, this would have made the conclusion of Theorem \ref{wh_ah_sep_thm} much less convincing.
\begin{corol}\label{wh_ah_sep_cor}
There exists a computable oracle $O$ such that for every $t\geq 2$, $\mathsf{W}[t]^{O}\subsetneq \mathsf{A}[t]^{O}$.
\end{corol}
Finally, we show that the $\mathsf{W}$-Hierarchy is not likely to be contained in any finite level of the $\mathsf{A}$-Hierarchy.
\begin{theor}\label{wh_ah_sep_thm_2}
There exists for each $t\geq 1$ a computable oracle $O_{t}$ such that $\mathsf{W}[t+1]^{O_t}\not\subset \mathsf{A}[t]^{O_t}$, where both machines have tail-restricted access to $O_t$, but the $\mathsf{W}[t]$-machine has the stronger type of oracle access mentioned in Section \ref{wp_fpt_sec}.
\end{theor}

As mentioned in the introduction, each oracle $O_t$ also separates the classes $\mathsf{W}[t]$ and $\mathsf{W}[t+1]$ in the relativized setting, since $\mathsf{W}[t]^{O_t}\subseteq\mathsf{A}[t]^{O_t}$ but $\mathsf{W}[t+1]^{O_t}\not\subset \mathsf{A}[t]^{O_t}$.

Although the conclusion of Corollary \ref{wh_ah_sep_cor} is made more believable by the use of the oracle $\mathsf{W}[t]$-machines described in Section \ref{wp_fpt_sec} (since the separation is against a more powerful oracle machine), Theorem \ref{wh_ah_sep_thm_2} would have been more convincing if the oracle machines did not have an enhanced ability to combine nondeterminism with oracle queries. On the other hand, it seems that if one further weakens the oracle access, it may not be possible to prove that $\mathsf{W}[t]^{\mathsf{FPT}}=\mathsf{W[P]}$, which this is a very reasonable identity in light of Theorem \ref{fpt_low_thm}.

\section{Conclusion and open problems}
Our results, together with the previously known theorems mentioned in the introduction, strongly indicate that if the central conjecture of parameterized complexity theory holds at all, proving it may be hard \emph{even under the additional assumption of a separation between arbitrarily-weakly-nondeterministic polynomial-time and $\mathsf{P}$} (and, in particular, that $\mathsf{P}\neq\mathsf{NP}$). Of course, the same also applies to the nowadays ``standard'' conjecture ETH. Additionally, we have seen that $\mathsf{W}[1]$ and $\mathsf{FPT}$ are in some ways unexpectedly close, unless much of what is generally assumed in parameterized complexity theory (such as the $\mathsf{W}$-Hierarchy not collapsing) is false. All of this suggests that the hardness of a problem for up to $\mathsf{W[P]}$ should not be treated as strong evidence that the problem is computationally intractabile, at least not with a similar level of confidence as when $\mathsf{NP}$-hardness is considered evidence of intractability.\\
\\
We mention some open problems:\\

\noindent\textbf{Is $\mathsf{W[P]}=\mathsf{W}[1]$?} A recurring issue in Section \ref{fpt_w1_sec} has been whether $\mathsf{W[P]}$ and $\mathsf{W}[1]$ might be equal. This is an interesting possibility for at least two reasons. First, unlike the case of $\mathsf{W[P]}=\mathsf{FPT}$, a collapse of $\mathsf{W[P]}$ to $\mathsf{W}[1]$ could, as far as we know, be proved \emph{without} any new insight into the extent to which $\mathsf{P}$ can simulate limited nondeterminism. Second, a collapse of $\mathsf{W[P]}$ to $\mathsf{W}[1]$ would have the effect of greatly simplifying the landscape of parameterized complexity: the entire $\mathsf{W}$-Hierarchy would collapse to $\mathsf{W}[1]$, and the $\mathsf{A}$-Hierarchy would coincide, level-by-level, with $\mathsf{W[P]H}$, and therefore exhibit many of the properties we know to hold for the Polynomial Hierarchy (including downward separation).

At the very least, it would be interesting to know if such a collapse would have any effect on the classical complexity world.\\

\noindent\textbf{Downward separation for the parameterized hierarchies.} We have seen evidence that the levels of the $\mathsf{W}$- and the $\mathsf{A}$-Hierarchy are in various ways not far apart. In particular, the entire $\mathsf{W}$-Hierarchy appears to be closer to $\mathsf{FPT}$ than any class defined in terms of limited nondeterminism is to $\mathsf{P}$. Beigel and Goldsmith \cite{beigol} have shown that for the $\mathsf{\beta}$-Hierarchy (whose levels are the classes $\beta_i\mathsf{P}:=\mathsf{NP}[(\log n)^i]$), downward separation fails in a relativized setting, in the sense that they can construct oracles relative to which any finite set of collapses occurs without entailing further collapses.

Given the apparent closeness of the levels of the parameterized hierarchies to each other and to $\mathsf{FPT}$, and that for the classes $\beta_i\mathsf{P}$, downward separation can be made to fail relative to some oracles, we conjecture that it is possible to construct oracles relative to which downward separation fails for the $\mathsf{W}$- and the $\mathsf{A}$-Hierarchy. More precisely, we conjecture that there exist computable parameterized oracles $O_1$ and $O_2$ such that:
\begin{align*}
&\mathsf{W}[1]^{O_1}=\mathsf{W}[2]^{O_1}\textrm{ but }\mathsf{W}[2]^{O_1}\neq\mathsf{W}[3]^{O_1}\textrm{, and}\\
&\mathsf{A}[1]^{O_2}=\mathsf{A}[2]^{O_2}\textrm{ but }\mathsf{A}[2]^{O_2}\neq\mathsf{A}[3]^{O_2}.
\end{align*}

\subsection*{Acknowledgments}
I thank Harry Buhrman, S\'andor Kisfaludi-Bak, and Ronald de Wolf for helpful discussions. The counter-example in the proof of Proposition \ref{counter_example_prop} is due to Harry Buhrman. I am especially grateful to Ronald de Wolf and Leen Torenvliet for helpful comments on drafts of the paper.

\bibliography{rcbott_201712.bib}

\begin{thebibliography}{10}

\bibitem{allender91}
E.~Allender.
\newblock Limitations of the upward separation technique.
\newblock {\em Mathematical Systems Theory}, 24(1):53--67, 1991.

\bibitem{ab}
S.~Arora and B.~Barak.
\newblock {\em Computational {C}omplexity: {A} {M}odern {A}pproach}.
\newblock Cambridge, 2009.

\bibitem{beigol}
R.~Beigel and J.~Goldsmith.
\newblock Downward separation fails catastrophically for limited nondeterminism
  classes.
\newblock {\em SIAM J. Comput.}, 27(5):1420--1429, 1994.

\bibitem{bookwilxu}
R.V. Book, C.B. Wilson, and M.~Xu.
\newblock Relativizing time, space, and time-space.
\newblock {\em SIAM J. Comput.}, 11(3):571--581, 1982.

\bibitem{rcb17}
R.C. Bottesch.
\newblock {Relativization and Interactive Proof Systems in Parameterized
  Complexity Theory}.
\newblock In {\em 12th International Symposium on Parameterized and Exact
  Computation (IPEC 2017)}, volume~89, pages 9:1--9:12, 2018.

\bibitem{bussislam06}
J.F. Buss and T.~Islam.
\newblock Simplifying the {W}eft hierarchy.
\newblock {\em Theoretical Computer Science}, 351(3):303--313, 2006.

\bibitem{caichen95}
L.~Cai, J.~Chen, R.G. Downey, and M.R. Fellows.
\newblock On the structure of parameterized problems in {N}{P}.
\newblock {\em Information and Computation}, 123:38--49, 1995.

\bibitem{cfg3}
Y.~Chen, J.~Flum, and M.~Grohe.
\newblock Machine-based methods in parameterized complexity theory.
\newblock {\em Theoretical Computer Science}, 339:167--199, 2005.

\bibitem{df92}
R.G. Downey and M.R. Fellows.
\newblock Fixed-parameter tractability and completeness {I}{I}{I} - {S}ome
  structural aspects of the {W} hierarchy.
\newblock In K.~Ambos-Spies, S.~Homer, and U.~Schoning, editors, {\em
  Complexity Theory}, pages 166--191. Cambridge University Press, 1993.

\bibitem{df1}
R.G. Downey and M.R. Fellows.
\newblock {\em Parameterized {C}omplexity}.
\newblock Springer, Berlin, 1999.

\bibitem{df2}
R.G. Downey and M.R. Fellows.
\newblock {\em Fundamentals of {P}arameterized {C}omplexity}.
\newblock Springer, 2013.

\bibitem{fg}
J.~Flum and M.~Grohe.
\newblock {\em Parameterized {C}omplexity {T}heory}.
\newblock Springer, Berlin, 2006.

\bibitem{fortsip}
L.~Fortnow and M.~Sipser.
\newblock Are there interactive protocols for co-{N}{P} languages?
\newblock {\em Information Processing Letters}, 28(5):249--251, 1988.

\bibitem{gavtoren}
R.~Gavalda, L.~Torenvliet, O.~Watanabe, and J.L. Balcazar.
\newblock Generalized {K}olmogorov complexity in relativized separations.
\newblock {\em Mathematical Foundations of Computer Science (MFCS). Lecture
  Notes in Computer Science}, 452:269--276, 1988.

\bibitem{ipz}
R.~Imagliazzo, R.~Paturi, and F.~Zane.
\newblock Which problems have strongly exponential complexity?
\newblock {\em Journal of Computer and System Sciences}, 63(4):512--530, 2001.

\bibitem{lokmarsau}
D.~Lokshtanov, D.~Marx, and S.~Saurabh.
\newblock Lower bounds based on the exponential time hypothesis.
\newblock {\em Bulletin of the EATCS}, 105:41--71, 2011.

\bibitem{papadim}
C.H. Papadimitriou.
\newblock {\em Computational {C}omplexity}.
\newblock Addison-Wesley, 1994.

\bibitem{sham}
A.~Shamir.
\newblock {IP} = {PSPACE}.
\newblock {\em J. ACM}, 39(4):869--877, 1992.

\end{thebibliography}

\section*{Appendix}
\subsection*{Proofs of theorems in Section \ref{wp_fpt_sec}}
For the theorems in this section, complexity classes are defined in terms of TMs whenever possible (this applies to $\mathsf{P}$, $\mathsf{FPT}$, $\mathsf{NP}$, and $\mathsf{W[P]}$ in particular), instances of problems and oracles query instances are encoded as finite sequences of symbols from the set $\{0,1,\#\}$, and oracles are classical (not parameterized).
\begin{lemma}
Let $h:\mathbb{N}\rightarrow\mathbb{N}\cup\{0\}$ be an unbounded non-decreasing function such that $h(n)\leq\log(n)$ for all $n\in\mathbb{N}$. Then for any $i\in\mathbb{N}$ there exist infinitely many values $n$ such that $h(n)^i\geq h(n^i)$.
\end{lemma}
\begin{proof}
For $i=1$ the desired conclusion holds trivially. Let $h$ be as above and assume that for some $i\geq 2$ the conclusion does not hold, meaning that there exists an $n_0\in\mathbb{N}$ such that 
\begin{equation}\label{fct_lemma_eq}
\forall n\geq n_0:h(n)^i<h(n^i).
\end{equation}
Let $n_1\geq n_0$ be such that $h(n_1)\geq 2$ (such a value exists because $h$ is unbounded). By repeated applications of (\ref{fct_lemma_eq}) we get that for all $k\geq 1$, $h(n_1^{i^k})=h((n_1^{i^{k-1}})^i)>h(n_1^{i^{k-1}})^i>h(n_1^{i^{k-2}})^{i^2}>\ldots>h(n_1)^{i^k}\geq 2^{i^k}$. On the other hand, we have by the assumed bound on $h$ that $h(n_1^{i^k})\leq \log(n_1^{i^k})=i^k\log(n_1)$. Combining the two inequalities, we get that $i^k\log(n_1)\geq 2^{i^k}$, which leads to a contradiction for $k$ sufficiently large.
\end{proof}

\begin{proof}[Proof of Theorem \ref{wp_fpt_thm}]
Without loss of generality we may assume that $h$ is the quasi-inverse of a time-constructible function $g$ such that $g(n)\geq 2^n$ for all $n\in\mathbb{N}$ (see \cite[Lemma 1.35 and Lemma 3.24]{fg}).
Let $N_1,N_2\ldots$ be a computable list of all nondeterministic TMs. For a given set $A\subseteq\{0,1\}^\ast$, we define:
\begin{align*}
Q(A):=\{&y\mid y=x\#1^k\#0^{g(2^k) n^4}\#i,\textrm{ where }n=|x|+k+1\textrm{ and }N^{A\oplus Q(A)}_i\textrm{ accepts $x\#1^k$}\\
&\textrm{within $kn$ steps, using at most $k\log n$ nondeterministic bits.}\}
\end{align*}

The reason why we add $g(2^k)(|x|+k+1)^4$ zeroes to the strings in $Q(A)$ is to make querying this part of the oracle expensive for polynomial-time deterministic TMs, as well as for $\mathsf{W[P]}$-machines running in time $k(|x|+k+1)$, the latter being necessary in order for the recursive definition of $Q(A)$ to not be circular. The proof now proceeds as follows: First, we show that $Q(A)$ is well-defined. Next, we prove that, regardless of the set $A$, $\mathsf{FPT}$ and $\mathsf{W[P]}$ coincide relative to $A\oplus Q(A)$. Finally, we construct $A$ so as to separate $\mathsf{NP}[h(n)\log n]$ from $\mathsf{P}$ relative to $A\oplus Q(A)=:B$.\\

\emph{1) $Q(A)$ is well-defined, and computable if $A$ is computable.}

We show that $Q(A)$ can be constructed in stages in such a way that, whenever we decide whether to place a particular $y$ into $Q(A)$, the computation whose outcome is encoded by $y\in Q(A)$ can only query instances that have already been decided at an earlier stage. By induction on $m=|x|+k$, we may assume that $y\in Q(A)$ has been decided for all strings of the form $x\#1^k\#0^{g(2^k)(|x|+k+1)^4}\#i$, with $|x|+k<m$ (this is trivially true when $m=1$). For every $y$ with $|x|+k=m$, we have $|y|>(|x|+k+1)^4>k(|x|+k+1)$, whereas the corresponding computation can only query strings of length $<k(|x|+k+1)$ (due to the bound on its running time). Thus we have that a computation encoded by a string $y$ whose membership in $Q(A)$ is decided at stage $m$, can not query whether $y'\in Q(A)$ for any $y'$ for which $|x'|+k'\geq m$. We may conclude that $Q(A)$ is well-defined for any set $A$, and computable if $A$ is.\\

\emph{2) For any set $A\subseteq\{0,1\}^\ast$ we have $\mathsf{FPT}(A\oplus Q(A))=\mathsf{W[P]}(A\oplus Q(A))$.}

Let $R\subseteq\{x\#1^k\mid x\in\{0,1\}^\ast,k\in\mathbb{N}\}$ be a problem decided by some $\mathsf{W[P]}$-machine with unrestricted access to the $A\oplus Q(A)$ oracle. Then there exist $i\in\mathbb{N}$, a computable function $f$, and a polynomial function $p$, such that on input $x\#1^k$, $N^{A\oplus Q(A)}_i$ decides in time $f(k)p(|x|+k)$ whether $x\#1^k\in R$, using at most $f(k)\log(|x|)$ nondeterministic bits. Then there is a $j\in\mathbb{N}$ such that on input $x'\#1^{k''}$, $N^{A\oplus Q(A)}_j$ verifies that $x'$ is of the form $0^t10^u10^{k'}10^k1x$, with $x\in\{0,1\}^\ast$ (by counting the number of zeroes in each of the four sequences) and then performs the computation of $N^{A\oplus Q(A)}_i$ on input $x\#1^k$ for at most $k'u$ steps and using at most $k'\log u$ nondeterministic bits, and accepts if and only if the simulated computation does.

The first phase takes time $O(t+u+k'+k+|x|)$ if efficient counters are used. Simulating the computation of $N^{A\oplus Q(A)}_i$ with counters for the number of steps and the nondeterministic guesses takes time at most some polynomial in $k'u$. Thus, if $t=u^c$ and $k''=(k')^c$, for a suitable constant $c>1$, then $N^{A\oplus Q(A)}_j$ will run in time $k''t$, use fewer than $k''\log|x'|$ nondeterministic guesses, and accept if and only if $N^{A\oplus Q(A)}_i$ accepts in the right amount of time, with the right number of nondeterministic steps.

Finally, a deterministic machine can query $Q(A)$ to check whether $N^{A\oplus Q(A)}_j$ accepts on input $x'\#1^{f(k)^c}$, where $x'=0^{p(|x|+k)^c}10^{p(|x|+k)}10^{f(k)}10^k1$. Clearly, such a query can be constructed in fpt-time.\\

\emph{3) Diagonalization against $\mathsf{P}$-machines.}

Let $L$ be the language $\{0^n \mid n\in\mathbb{N}\textrm{ and }\exists y\{0,1\}^{\lceil h(n)\log(n)\rceil}\textrm{: } y\in A\}$. Clearly we have that $L\in\mathsf{NP}[h(n)\log n]^{A}\subseteq\mathsf{NP}[h(n)\log n]^{A\oplus Q(A)}$.

Let $P_1,P_2,\ldots$ be a computable enumeration of polynomial-time oracle TMs, in which every such machine appears infinitely many times. As is usual in such arguments, we consider a process in which each $P_i^{A\oplus Q(A)}$ is run on input $0^n$ for $n^i$ steps, and the set $A$ is defined on-the-fly so that every machine gives the wrong answer on at least one input (with respect to the question whether that input is in $L$).

Although we only decide which strings to place into $A$, this implicitly determines the strings in $Q(A)$, in a way that we do not control. Since every $P_i$ has access to both $A$ and $Q(A)$, the latter part of the oracle can also influence the outcome of computations. Furthermore, the strings in $Q(A)$ encode the outcomes of computations performed by nondeterministic machines with access to the oracle $A\oplus Q(A)$, and these computations also depend on $A$, as well as on the outcomes of other nondeterministic oracle computations, and so on. In order to ensure that we can fool each polynomial-time TM on some input, we need to upper-bound the number of strings whose membership or non-membership in $A$ can directly or indirectly influence the deterministic oracle computation.

At the level of the deterministic computation of a machine $P_i^{A\oplus Q(A)}$ on input $0^n$, the machine can make up to $n^i$ queries to $A$ or to $Q(A)$. Let $l:=n^i$. When $P_i$ queries $Q(A)$ for the outcome of a nondeterministic computation on an input $(x,k)$, the format of the corresponding query string imposes the restrictions $|x|\leq \sqrt[4]{l}$ and $k\leq \log h(l)$. These computations can then run for $k(|x|+k+1)\leq (\sqrt[4]{l})^2=\sqrt{l}$ steps and use $k\log(|x|+k+1)\leq 2\log h(l)\log \sqrt[4]{l}=(1/2)\log h(l)\log l$ nondeterministic bits (both for $l$ sufficiently large). Therefore, such a computation can make at most $\sqrt{l}$ queries to either $A$ or $Q(A)$ on each of at most $2^{(1/2)\log h(l)\log l}$ computation paths. If on one of these computation paths a query to $Q(A)$ is made for the outcome of a computation on an input $(x',k')$, then we have again that $|x'|\leq \sqrt[4]{\sqrt{l}}=\sqrt[8]{l}$, and we upper-bound $k'$ generously by $\log h(l)$ again. By the same reasoning we have that the computations on this second nondeterministic level can make at most $\sqrt[4]{l}$ queries on each of $2^{(1/4)\log h(l)\log l}$ computation paths. In this way we can upper-bound the number of $A$-queries that can be made by any non-deterministic computation whose outcome can be queried as a result of $Q(A)$-queries made by $P_i$. Note that with the diminishing limits on the running times, a query of ``order'' $\log l$ will refer to the outcome of a computation with a running time upper-bounded by $l^{1/l}< 2$, meaning that no non-trivial oracle queries can be made. An upper bound on the total number of queries (and hence, the number of queries to $A$) that influence the original deterministic computation, either directly or indirectly, can now be obtained by multiplying the upper-bounds for all the different levels of oracle queries. We get a bound of:
\begin{align*}
&l\cdot 2^{(1/2)\log h(l)\log l}\sqrt{l}\cdot 2^{(1/4)\log h(l)\log l}\sqrt[4]{l}\cdot\ldots\\
\leq & l^{1+1/2+1/4+\ldots}\cdot 2^{(1/2+1/4+\ldots)\log h(l)\log l}\leq 2^{(2+\log h(l))\log l}.
\end{align*}

Now we construct the set $A$ in the standard manner: At stage $i=1,2,\ldots$ we choose $n$ so that no queries to $A$ with strings of length $\geq h(n)\log(n)$ have been made at previous stages, and such that $h(n)\log n>i^2\log h(n)\log n+2i\log n$, and $h(n)^i\geq h(n^i)$ (that $n$ can be chosen in this way follows from the unboundedness of $h$ and Lemma \ref{fct_lemma_eq}). Then we have:
\[
h(n)\log n>i^2 \log h(n)\log n+2i\log n\geq \log h(n^i)\log(n^i)+2\log(n^i)=(2+\log h(l))\log l.
\]
We then simulate $P_i^{A\oplus Q(A)}$ on input $0^n$, including all the nondeterministic computations that need to be simulated as a result of queries to $Q(A)$. Whenever $A$ is queried with a new string for the first time, the answer is negative (in particular, this applies to all strings of length $\geq h(n)\log(n)$); all other queries are answered consistently with previous answers. By our computed upper bound and the choice of $n$, we have that the number of strings $y\in\{0,1\}^{\lceil h(n)\log(n)\rceil}$ that will be queried throughout the simulation is $<2^{h(n)\log n}$. If $P_i^{A\oplus Q(A)}(0^n)$ terminates and rejects, we place an unqueried string $y\in\{0,1\}^{\lceil h(n)\log n\rceil}$ into $A$; if the computation terminates and accepts, we place no strings of this form into $A$. Thus, if the computation terminates, it will incorrectly decide whether $0^n\in L$. Since every polynomial-time TM appears infinitely many times in the list, $n^i$ will eventually be sufficiently large for any given machine to terminate. We conclude that $L\notin \mathsf{P}(A\oplus Q(A))$.
\end{proof}
\begin{proof}[Proof of Corollary \ref{wp_fpt_cor}]
Take $B$ to be the oracle from Theorem \ref{wp_fpt_thm}. $\mathsf{W}[1](B)\subseteq \mathsf{A}[1](B)\subseteq \mathsf{W[P]}(B)$ is obvious.

Although each level of the $\mathsf{W}$-Hierarchy is a subset of $\mathsf{W[P]}$, it is not immediately obvious (at least, not from the way this inclusion is usually proved, namely via reductions), that the inclusion also holds in the presence of oracles. For a fixed $t\geq 2$, consider a $\mathsf{W}[t]$-machine $W$ with unrestricted access to $B$ such that, on input $(x,k)$, $W$ runs in time $f(k)p(|x|)$ with nondeterministic guess instructions only among the last $h(k)$ instructions, for some computable functions $f,h$, and a polynomial function $p$. In the nondeterministic phase of a computation, such a machine can make some number of existential guesses, followed by $t-1$ blocks of instructions, each containing at most $c$ nondeterministic guess instructions of a single type (existential or universal), where $c$ is a constant that depends only on the machine, and is independent of the input. A $\mathsf{W[P]}(B)$-machine can simulate such a computation by making the same oracle queries at any given point in the computation, using its own existential guesses to simulate the first block of existential guesses of $W$, and simulating the remainder of the computation for every possible outcome of the remaining nondeterminstic guesses of $W$. This produces an overhead of $O((f(k)p(|x|))^{c(t-1)})$ in the running time of the simulating machine, but since $c$ and $t$ are constants, the dependence of this term on $|x|$ is still polynomial, and hence the oracle $\mathsf{W[P]}$-machine runs in fpt-time. (Note that for this simulation, it is essential that the oracle access of the $\mathsf{W[P]}(B)$-machine be unrestricted (and in particular not parameter-bounded), since it may need to make a number of oracle queries that depends on $|x|$, because it is simulating every possible computational path after the first block of existential guesses of $W$.) We conclude that $\mathsf{W}[t](B)\subseteq \mathsf{W[P]}(B)$, for all $t\geq 2$.

Finally, except for $\mathsf{A}[1]$, the levels of the $\mathsf{A}$-Hierarchy are not known to be $\subseteq \mathsf{W[P]}$. However, a close look at the proof that $\mathsf{FPT}=\mathsf{W[P]}\Rightarrow \mathsf{FPT}=\mathsf{A}[1]=\mathsf{A}[2]=\ldots$ (see \cite{fg}), reveals that it still holds even if the machines have unrestricted oracle access to $B$. We therefore also have that for all $t\geq 2$, $\mathsf{FPT}(B)=\mathsf{W[P]}(B)\Rightarrow \mathsf{A}[t](B)=\mathsf{FPT}(B)$.
\end{proof}

\subsection*{Proofs of theorems in Section \ref{wh_ah_sec}}
For the theorems in this section, all complexity classes are defined in terms of random access machines, instances of problems and oracles query instances are encoded as sequences of positive integers, and oracles are parameterized.
\begin{proof}[Proof of Theorem \ref{wh_ah_sep_thm}]
Let $L(O)$ be the parameterized problem defined as follows:
\begin{displaymath}
\{(0^n,k)\mid n,k\in\mathbb{N},k\leq n,\textrm{ and }\forall i_1,\ldots,i_k\in [n]:((i_1,\ldots,i_k),k)\in O\}.
\end{displaymath}
A $\mathsf{co\textrm{-}A}[1]^O$-machine, whose oracle access is tail-restricted, can decide $L(O)$ by simply guessing a query instance of the appropriate form (a vector of $k$ integers from $[n]$, with parameter value $k$), querying the oracle, and accepting if and only if the answer is `yes'.

We show how a finite part of the oracle $O$ can be defined so that a $\mathsf{para\textrm{-}NP}$-machine with known running time bounds does not decide $L(O)$. Let $M$ be an oracle $\mathsf{para\textrm{-}NP}$-machine which on input $(x,k)$ runs in time $f(k)(|x|+k)^c$, where $f$ is a computable function and $c\geq 1$ is a constant. $M$ can make nondeterministic guesses and query $O$ throughout the computation. Let $k>c$ be fixed, and let $n\in\mathbb{N}$ be such that $n^k>f(k)(n+k)^c$. The machine is run on input $(0^n,k)$. Initially, all queries to vectors from the set $[n]^{k}$, on all computational paths, are answered affirmatively, while queries on any instance not of this form are answered in a well-defined way (which will be specified later). If $M$ rejects on all computational paths, we place all vectors from $[n]^k$ (with parameter value $k$) into $O$, making $(0^n,k)$ a `yes'-instance of $L(O)$. If, on the other hand, $M$ accepts the input on some computational path, it will do so without having queried all $n^k$ relevant vectors in on this path, and we can remove from $O$ one of the instances of this form that has not been queried, without changing the fact that $M$ nondeterministically (in an existential sense) accepts the input. But this change to $O$ makes $(0^n,k)$ a `no'-instance.

In order to apply the above procedure to all $\mathsf{para\textrm{-}NP}$-machines, we computably list all valid RAM programs and simulate each of them repeatedly for a bounded number of steps. The inputs on which a machine $M$ is run for its $i$-th simulation, are chosen as a computable function of $M$ and $i$, in such a way that for infinitely many $k\in \mathbb{N}$, $M$ is run on infinitely many inputs $(0^n,k)$ (in other words, so that $k$ grows arbitrarily large, and for infinitely many values of $k$, $n$ also grows arbitrarily large). For each new simulation of the timed computation of a machine $M$, we choose the parameter value $k\in\mathbb{N}$ as mentioned above (computably in terms of $M$ and the number of the simulation), and set $n=1+\max\{n',n'',k\}$, where $n'$ is the largest value such that $(0^{n'},k')$ was the input for some previous simulation (for any machine and any $k'$), and $n''$ is the largest value for which an oracle query of the form $(v,k'')$, with $v\in\mathbb{N}^{n''},k''\in\mathbb{N}$, was made during some simulation up to this point. $M$ is then simulated on input $(0^n,k)$ for $n^{k}$ steps, and all queries not of the form $(v,k)$, $v\in [n]^{k}$, are answered consistently with previous simulations, or with `yes' if they are first-time queries. If $M$ is indeed a $\mathsf{para\textrm{-}NP}$-machine, then for some computable function $f$ and constant $c\geq 1$, it runs in time $f(k)(|x|+k)^c$ and, for a sufficiently large values $n$ and $k$, we have $n^{k}>f(k)(n+k)^c$, and can therefore apply the procedure described in the previous paragraph to ensure that the machine nondeterministically accepts if and only if $(0^n,k)\notin L(O)$.
\end{proof}

\begin{proof}[Proof of Theorem \ref{wh_ah_sep_thm_2}]
For this separation we once again need an oracle whose elements are strings of non-negative integers. We sketch the proof for $t=2$, but it is easy to see that it generalizes for larger values of $t$.

We want to construct an oracle $O$ such that $\mathsf{W}[2]^{O}\not\subseteq \mathsf{A}[1]^{O}$, where both machines have tail-restricted access to the oracle (recall that this means that the $\mathsf{W}[2]$-machine's oracle registers are write-only, but that it can copy values from its guess registers there).

Given an oracle $O\subset\mathbb{N}^\ast\times \mathbb{N}$, we define:
\[L(O):=\{(0^n,k)\mid \exists x\in[n]^k\textrm{ s.t. }\forall y\in[n]: (xy,k)\in O\}.
\]

Clearly, $L(O)\in\mathsf{W}[2]^O$ for any oracle $O$. To construct an oracle relative to which this language is not in $\mathsf{A}[1]$, we once again use a construction in stages, where we simulate every (existential) nondeterministic RAM on instances of the form $(0^n,k)$ for $n^k$ steps, with increasing values of $k$ and $n$, such that for each of infinitely many values of $k$, the machine is simulated on input $(0^n,k)$ for infinitely many values of $n$. This means that if a particular RAM is in fact an $\mathsf{A}[1]$-machine, then, for some constant $c>1$ and computable functions $f$ and $h$, on input $(0^n,k)$ it will run in time $f(k)(n+k)^c$, making nondeterministic guesses and oracle queries only among the last $h(k)$ steps of the computation. Thus, for $k>c$, $n>h(k)$, and $n^k>f(k)(n+k)^c$, the machine will halt, and make fewer than $n$ queries on any one computation path. The remainder of the argument is as in the proof of Theorem \ref{wh_ah_sep_thm}: All new queries are answered with `yes', and if a particular machine terminates and rejects on all paths, then we place $(xy,k)$ for some $x\in[n]^k$ and all $y\in[n]$ into $O$ (since the $\mathsf{A}[1]$-machine has rejected despite all queries being answered with `yes', adding more instances to the oracle can not cause the machine to reject). On the other hand, if the machine terminates and accepts, we fix a single accepting computation path of the machine on the given input, and remove for \emph{each} $x\in[n]^k$ some string $xy$ that was \emph{not} queried on this accepting computation path. Thus, we ensure that for sufficiently large $n$, every $\mathsf{A}[1]$-machine will have one instance where it gives the incorrect answer.
\end{proof}

\end{document}